\documentclass[preprint,12pt,3p]{elsarticle}

\usepackage{amssymb}
\usepackage{graphicx}

\usepackage{pgf,tikz}

\usetikzlibrary{snakes}
\tikzstyle{printersafe}=[snake=snake,segment amplitude=0 pt]

\usetikzlibrary{arrows}
\usetikzlibrary{shapes}
\usepackage{amsmath}
\usepackage{caption}

\newtheorem{theorem}{Theorem}

\newtheorem{conjecture}{Conjecture}

\newenvironment{proof}{ {\bf Proof:}} {$\Box$}

\journal{Sample Journal}

\begin{document}

\begin{frontmatter}

\title{On two consequences of Berge-Fulkerson conjecture}

\author[label1]{Vahan V. Mkrtchyan\corref{cor1}}
\address[label1]{Dipartimento di Informatica,
Universita degli Studi di Verona, 37134 Verona, Italy}

\cortext[cor1]{The corresponding author}

\ead{vahanmkrtchyan2002@ysu.am}

\author[label5]{Gagik N. Vardanyan}
\address[label5]{Department of Informatics and Applied Mathematics, Yerevan State University, Yerevan, 0025, Armenia}
\ead{vgagik@gmail.com}


\begin{abstract}
The classical Berge-Fulkerson conjecture states that any bridgeless cubic graph $G$ admits a list of six perfect matchings such that each edge of $G$ belongs to two of the perfect matchings from the list. In this short note, we discuss two statements that are consequences of this conjecture. We show that the first statement is equivalent to Fan-Raspaud conjecture. We also show that the smallest counter-example to the second one is a cyclically $4$-edge-connected cubic graph.
\end{abstract}


\end{frontmatter}



\section{The main result}

For a bridgeless cubic graph $G$, and a list of (not necessarily distinct) perfect matchings $C=(F_1,...,F_k)$, let $\nu_C(e)$ be the number of perfect matchings of $C$ that contain the edge $e$. For non-defined concepts we refer to \cite{Zhang1997}. The main topic of this short note are the following two classical conjectures:

\begin{conjecture}
\label{conj:Berge} (Berge, unpublished) Let $G$ be a bridgeless cubic graph. Then there is a list $C=(F_1,...,F_5)$ of perfect matchings, such that $\nu_C(e)\geq 1$ for any edge $e$ of $G$.
\end{conjecture}

\begin{conjecture}
\label{conj:BergeFulkerson} (Berge-Fulkerson \cite{Fulkerson}) Let $G$ be a bridgeless cubic graph. Then there is a list $C=(F_1,...,F_6)$ of perfect matchings, such that $\nu_C(e)=2$ for any edge $e$ of $G$.
\end{conjecture} Clearly, Conjecture \ref{conj:BergeFulkerson} implies Conjecture \ref{conj:Berge}. In \cite{Mazz2010}, it is shown that Conjecture \ref{conj:Berge} implies Conjecture \ref{conj:BergeFulkerson}. Thus the two conjectures are equivalent. A list $C=(F_1, F_2, F_3)$ of a bridgeless cubic graph $G$ is called an FR-triple, if $F_1\cap F_2 \cap F_3=\emptyset$. The classical Fan-Raspaud conjecture asserts:

\begin{conjecture}
\label{conj:FR} (Fan-Raspaud \cite{FR1994}) Any bridgeless cubic graph admits an FR-triple.
\end{conjecture}

In this paper, we consider the following two conjectures:

\begin{conjecture}
\label{conj:GagikFRtriple} For any bridgeless cubic graph $G$, any edge $e\in E(G)$ and $i\in \{0,1,2\}$, there is an FR-triple $C$ of $G$, such that $\nu_C(e)=i$.
\end{conjecture} 

\begin{conjecture}
\label{conj:Vahanconj} Let $G$ be a bridgeless cubic graph, $e$ and $f$ be adjacent edges, $0\leq i,j\leq 2$ be two numbers with $1\leq i+j\leq 3$. Then $G$ contains an FR-triple $C$, such that $\nu_C(e)=i$ and $\nu_C(f)=j$.
\end{conjecture}

It is easy to see that Conjecture \ref{conj:BergeFulkerson} implies Conjecture \ref{conj:Vahanconj}, Conjecture \ref{conj:Vahanconj} implies Conjecture \ref{conj:GagikFRtriple}, and Conjecture \ref{conj:GagikFRtriple} implies Conjecture \ref{conj:FR}. We are ready to obtain our first result.

\begin{theorem}
\label{thm:GagikFRthm} Conjecture \ref{conj:GagikFRtriple} is equivalent to Fan-Raspaud conjecture (Conjecture \ref{conj:FR}).
\end{theorem}

\begin{proof}
It suffices to show that Conjecture \ref{conj:FR} implies Conjecture \ref{conj:GagikFRtriple}. Let $G$, its edge $e$ and $0\leq i\leq 2$ be given. Take one copy of Petersen graph $P$, 15 copies $G_1,...,G_{15}$ of $G$, and let $e_1,...,e_{15}$ be the copies of $e$ in these 15 graphs. Consider a bridgeless cubic graph $H$ obtained as follows: for $j=1,...,15$ remove $j$th edge from $P$ and connect it with a $2$-edge-cut to $G_j-e_j$.

Now, by Conjecture \ref{conj:FR} $H$ admits an FR-triple $C_H$. Observe that $C_H$ gives rise to an FR-triple $C_P$ of $P$. Now, there is an edge $f$ of $P$ such that $\nu_{C_P}(f)=i$ (this is easy to prove by considering cases $i=0,1,2$ separately). Consider the FR-triple of $G$ corresponding to the copy $f$ arising from $C_H$. Observe that in this FR-triple $e$ belongs to exactly $i$ of the perfect matchings. The proof is complete. 
\end{proof}

Our next result deals with the smallest counter-example to Conjecture \ref{conj:Vahanconj}.

\begin{theorem}
\label{thm:VahanconjthmCyc4} The smallest counter-example to Conjecture \ref{conj:Vahanconj} is cyclically $4$-edge-connected.
\end{theorem}

\begin{proof} Let $G$ be a smallest counter-example to Conjecture \ref{conj:Vahanconj}. Clearly, it is connected. Let us show that it has no $2$-edge-cuts. On the opposite assumption, consider a $2$-edge-cut $C$. Let us show that $C\cap \{e,f\}\neq \emptyset$. If none of $e$ and $f$ is in $C$, then consider the standard two smaller graphs $G_1$ and $G_2$. Assume that $G_1$ contains both of the edges. Then since $G_1$ is smaller, it has an FR-triple containing $e$ and $f$, $i$ and $j$ times, respectively. Let $k$ be the frequency of the new edge of $G_1$ arising from $C$. Now, we can take a similar FR-triple in $G_2$, where the new edge of $G_2$ is covered exactly $k$ times, and if we glue these two FR-triples, we will have an FR-triple of $G$ containing $e$ and $f$, $i$ and $j$ times, respectively. This is a contradiction.

Now, assume that $e\in C$. Clearly $f\notin C$. Again assume that $G_1$ contains the edge $f$. Since $G_1$ is smaller, we have that it admits an FR-triple containing $e$ and $f$, $i$ and $j$ times, respectively. Now, take an FR-triple containing the new edge of $G_2$ $i$ times. If we glue these FR-triples, we will have an FR-triple containing $e$ and $f$, $i$ and $j$ times, respectively. This is a contradiction.

Thus, we have that $G$ is $3$-edge-connected. Now, let us show that all $3$-edge-cuts of $G$ are trivial. Assume that $G$ contains a non-trivial $3$-edge-cut $C$. Let us show that $C\cap \{e,f\}\neq \emptyset$. If none of $e$ and $f$ is in $C$, then consider the standard two smaller graphs $G_1$ and $G_2$. Assume that $G_1$ contains both of the edges. Then since $G_1$ is smaller, it has an FR-triple containing $e$ and $f$, $i$ and $j$ times, respectively. Let the frequency of edges of $C$ in this FR-triple of $G_1$ be $k_1, k_2, k_3$. Clearly, $k_1+k_2+k_3=3$ and $k_1+k_2\geq 1$, $k_1, k_2\leq 2$. Now, since $G_2$ is not a counter-example, we can find an FR-triple in $G_2$ such that the same two edges of $C$ are covered $k_1$ and $k_2$ times. Observe that the third edge has to be covered exactly $k_3$ times. Now, if we glue back these two FR-triples, we will have an FR-triple of $G$ covering $e$ and $f$, $i$ and $j$ times, respectively. This is a contradiction.

Thus, we are left with the case when $e\in C$. Observe that since $G$ is $3$-edge-connected, $C$ is a matching, hence $f\notin C$. Then since $G_1$ is smaller, it has an FR-triple containing $e$ and $f$, $i$ and $j$ times, respectively. Let the frequency of edges of $C$ in this FR-triple of $G_1$ be $k_1, k_2, k_3$. Clearly, $k_1+k_2+k_3=3$ and $k_1+k_2\geq 1$, $k_1, k_2\leq 2$. Now, since $G_2$ is not a counter-example, we can find an FR-triple in $G_2$ such that the same two edges of $C$ are covered $k_1$ and $k_2$ times. Observe that the third edge has to be covered exactly $k_3$ times. Now, if we glue back these two FR-triples, we will have an FR-triple of $G$ covering $e$ and $f$, $i$ and $j$ times, respectively. This is a contradiction. 

Thus, all $3$-edge-cuts of $G$ have to be trivial, which means that $G$ is cyclically $4$-edge-connected. The proof is complete.
\end{proof}

\medskip

Let us note that it is unknown whether the smallest counter-example to Fan-Raspaud conjecture is cyclically $4$-edge-connected. On the other hand, it can be shown that the smallest counter-example to Conjecture \ref{conj:GagikFRtriple} is $3$-edge-connected. It would be interesting to show that Conjecture \ref{conj:Vahanconj} is equivalent to Conjecture \ref{conj:FR}.

\bibliographystyle{elsarticle-num}


\end{document}